\documentclass[prl,twocolumn,superscriptaddress,footinbib]{revtex4}
\usepackage{enumerate}
\usepackage{amsmath,amssymb,latexsym}
\usepackage{amsthm} 
\usepackage{ascmac}
\usepackage{bm}
\usepackage{enumitem}
\usepackage{natbib}
\usepackage{hyperref}
\usepackage{braket}
\usepackage{ulem}

\def\II{\mathbb I}

\def\U#1{{\rm #1}}

\newtheorem{lemma}{{\bf Lemma}}

\def\Pr{\U{Pr}}
\def\tr{\U{tr}}

\def\({\left(}
\def\){\right)}

\begin{document}
\title{Security loophole in error verification in quantum key distribution}

\author{Toyohiro Tsurumaru
\footnote{Tsurumaru.Toyohiro@da.MitsubishiElectric.co.jp}}
\affiliation{Mitsubishi Electric Corporation, Information Technology R\&D Center, 5-1-1 Ofuna, Kamakura-shi, Kanagawa 247-8501, Japan
}

\author{Akihiro Mizutani
\footnote{mizutani@eng.u-toyama.ac.jp}
}
\affiliation{Faculty of Engineering, University of Toyama, Gofuku 3190, Toyama 930-8555, Japan}

\author{Toshihiko Sasaki
\footnote{Toshihiko.Sasaki@quantinuum.com}}
\affiliation{Quantinuum K.K. Otemachi Financial City Grand Cube 3F, Global Business Hub Tokyo 1-9-2 Otemachi, Chiyoda-ku, Tokyo 100-0004 Japan}

\begin{abstract}
  {
  The security of quantum key distribution (QKD) is evaluated based on the secrecy of Alice's key and the correctness of the keys held by Alice and Bob. A practical method for ensuring correctness is known as error verification, in which Alice and Bob reveal a portion of their reconciled keys and check whether the revealed information matches. 
   In this paper, we point out that when error verification is performed in a QKD protocol, the definition of secrecy must be revised accordingly. We illustrate the necessity of this revision with a counterexample, showing that neglecting it can lead to an incorrect security claim. In particular, we observe that in the case of security proof method based on phase error correction, which is one of the mainstream approaches and also known as Koashi's approach, no explicit method has been established to properly incorporate the revised secrecy definition.
To resolve this issue, we present a way to translate the phase error correction-based approach into another mainstream approach, called the leftover hashing lemma-based approach, also known as Renner's approach, where a solution has already been formulated. As a consequence, security proofs under the phase error correction-based approach automatically remain valid without any change in the secret key length, even if they implicitly consider error verification without revising the secrecy definition.
}
\end{abstract}

\maketitle

\section{Introduction}
The standard goal of security proofs of quantum key distribution (QKD)~\cite{Lo2014,RevModPhys.92.025002,Pirandola:20} is to derive the security parameter defined based on the universal composable security framework~\cite{RevModPhys.94.025008,benor2004,Muller-Quade_2009}. The security parameter is, roughly speaking, defined as the trace distance between the ideal secret keys and the actual keys (see Sec.~\ref{sec:usual_separation} for its definition). Toward this goal, it is customary and convenient first to split the security parameter into the secrecy and the correctness parameters, and then to derive each of them separately~\cite{Koashi_2009,Hayashi_2012,Tomamichel_2012,Tomamichel2017largelyself}.
One of them, the correctness parameter is defined by the probability that Alice's and Bob's secret keys are not identical.
The prevalent method for deriving this parameter is called {\it error verification} (see, for example, Ref.~\cite{tupkary2025} for details),
wherein Alice and Bob publicly compare hash values of their reconciled keys (i.e., the keys  obtained after completing bit error correction) to check the identicalness of these keys.
While other methods may in principle be able to serve for the same purpose~\cite{footnote1}, error verification is widely used~\cite{renner2006securityquantumkeydistribution,Tomamichel_2012,Curty2014,PhysRevA.89.022307,Tomamichel2017largelyself,Maeda2019,Sandfuchs2025securityof,mizutani2025protocolleveldescriptionselfcontainedsecurity,kamin2025} because it is by far the simplest and most reliable method in practice.
In this respect, it is an essential part of practical QKD implementations.

When considering the use of the secret key generated by QKD in subsequent cryptographic applications, we point out that the outcome of error verification, denoted by $V$, must be publicly announced (see Sec.~\ref{sec:Verifications_outcome_must_be_announced} for details).
Once this public nature of $V$ is accepted, one can readily conclude that the secrecy parameter must be defined for the state {\it after} error verification 
(see Sec.~\ref{sec:Separation_lemma_with_verification} for details).

However, security proofs based on the phase-error-correction approach~\cite{koashi2005simple,Koashi_2009}, also known as Koashi's approach, appear to inappropriately treat $V$ as secret information and instead define secrecy for the state {\it without} error verification~\cite{Takesue2015, PhysRevResearch.5.023132,PhysRevResearch.5.023065,Mizutani2015,Tamaki_2018,Curras2021,Curras2025,rrdps2014}. This is the core of the problem concerning the treatment of the outcome of error verification, which we identify as the central problem of this paper and refer to as the {\it verification problem}.

As one of the main contributions of this paper, we demonstrate the serious consequences of this inappropriate definition of secrecy by presenting a counterexample. In this example, a false claim of security can be made under this definition for a state without error verification, even though the claim does not hold in reality (see Sec.~\ref{subsec:conterexample} for detail). This situation occurs because in a certain type of protocol, the one bit of information $V$ may become correlated with the secret key, thereby compromising the security.

It should be noted that Koashi's approach (as summarized in Sec.~\ref{sec:summarykoashidifficulty}) has been the mainstream method for security proofs (see, for example, Sec.~II B 3 of the review paper~\cite{RevModPhys.92.025002}) and has been applied to a wide variety of QKD protocols, such as round-robin DPS~\cite{rrdps2014,Takesue2015,rrdpsexperiment1,rrdpsexperiment2}, decoy-state BB84~\cite{koashidecoy}, BB84 with an uncharacterized source~\cite{koashipreskill}, six-state protocol~\cite{katotamakisix}, twin field protocol~\cite{Maeda2019}, loss tolerant protocol~\cite{losstoleranttamaki}, DPS protocol~\cite{dpsmizutaninpj} and continuous variable protocols~\cite{matsuuraCV,matsuuraCV2}.
Nevertheless, the verification problem persists and leads to a security flaw~\cite{Takesue2015, PhysRevResearch.5.023132,PhysRevResearch.5.023065,Mizutani2015,Tamaki_2018,Curras2021,Curras2025,rrdps2014}. We also explain in Sec.~\ref{sec:summarykoashidifficulty} why this problem is difficult to resolve within the framework of Koashi's approach. 
Note that even if a protocol does not explicitly describe a verification step, there is no practical way to guarantee correctness other than by employing error verification. Consequently, as long as the security proof relies on Koashi's approach, the verification problem inevitably arises.

Fortunately, in Renner's approach~\cite{renner2006securityquantumkeydistribution}, which is another mainstream method for QKD security proofs, the verification problem has already been solved~\cite{Tomamichel2017largelyself}. We discuss this in Sec.~\ref{sec:Solution_Renner}. Here, Renner's approach is based on the leftover hashing lemma~\cite{renner2006securityquantumkeydistribution} for min-entropy, and we can relate the min-entropy of the quantum state after error verification to that without error verification (specifically, by using Lemma~10 in Ref.~\cite{Tomamichel2017largelyself}). As a result, even when the outcome of error verification is publicly revealed, secrecy can still be guaranteed without shortening the key length by even a single bit.

Another main contribution of this paper, in addition to clarifying the verification problem, is that we provide a simple solution to the problem in Koashi's approach. Specifically, we prove that security proofs based on Koashi's approach can always be repaired without reducing the final key length (see Sec.~\ref{sec:proposedsolution} for details). The basic idea is to translate Koashi's approach into Renner's approach by exploiting their equivalence established in~\cite{Tsurumaru_2020,9606732}, thereby resolving the problem within the framework of Renner's approach.

From a future perspective, the widespread adoption of QKD in society requires the standardization of a comprehensive framework for certifying its security. Our work represents an important contribution in this direction, as it clearly demonstrates the importance of rigorously incorporating error verification into the security proof and provides a practical method to address this challenge when one adopts Koashi's approach for the security proofs.

\section{Conventional argument of the separation}
\label{sec:usual_separation}
We begin by summarizing the notation adopted throughout this paper.

\begin{enumerate}
    \item    
    $[b]$ denotes the projector $\ket{b}\bra{b}$, with $\{\ket{b}\}_b$ being the computational basis. 
    \item 
    For a composite system described by a density operator $\rho_{AB\cdots}$ over multiple systems ($AB\cdots$), the state of a particular system (e.g., $\rho_A$) is defined by taking the partial trace over the remaining systems.
    \item 
    Given a quantum classical (sub normalized) state $\rho_{AB}$ of systems $AB$, $\rho_{A}^{B=b}$ is defined by $\tr_B[\rho_{AB}(\II_A\otimes[b]_B)]$. 
    \item 
    Given a density matrix $\sigma$, its trace norm is defined by~\cite{Nielsen_Chuang_2010}
    \begin{equation}
    \|\sigma\|_1:=\tr\sqrt{\sigma\sigma^\dagger}.
    \end{equation}
\end{enumerate}

In this section, we revisit the conventional argument for decomposing QKD's security parameter into those of secrecy and correctness based on Ref.~\cite{Koashi_2009}. This argument states that if Alice's final key is $\varepsilon_{\rm sec}$-secret and Alice's and Bob's final keys are $\varepsilon_{\rm cor}$-correct, then their pair of final keys as a whole satisfies $\varepsilon_{\rm sec}+\varepsilon_{\rm cor}$-security.

The more precise explanation would be as follows.
In this section, we restrict ourselves to the types of QKD protocols where decisions of continuing or aborting the protocol are made solely based on public information and do not depend on the contents of the sifted, reconciled, or secret keys. This situation typically arises in certain types of the BB84 protocol, where Alice and Bob abort the protocol if the estimated quantum bit error rate (QBER) during the sampling and parameter estimation phases exceeds a predetermined threshold. However, once the key distillation process---including error correction and privacy amplification---has commenced, they never abort the protocol.

In such cases, the security of a QKD protocol is defined as follows.
Let $K_A, K_B$ be the states of Alice's and Bob's secret keys and $E$ Eve's quantum system.
Also, let $\rho_{K_AK_BE}$ be the marginal (thus possibly sub-normalized) state corresponding to the event where the protocol is continued.
Then we say that the QKD protocol is $\varepsilon$-secure if
\begin{align}
\frac12
    \left\|
    \rho_{K_AK_BE}-\rho_{K_AK_BE}^{\rm ideal}
    \right\|_1
    \le\varepsilon    \label{eq:def_security_without_veri}
\end{align}
is satisfied, with the ideal state being
\begin{align}
\rho_{K_AK_BE}^{\rm ideal}=\sum_{k\in\{0,1\}^\ell}2^{-\ell}[k]_{K_A}\otimes[k]_{K_B}\otimes\rho_E,
\end{align}
and $\ell$ the length of the secret key. To prove Eq.~(\ref{eq:def_security_without_veri}), it is common to decompose the trace distance into two parameters (secrecy and correctness) and evaluate each separately. Specifically, the $\varepsilon_{\rm sec}$-secrecy of Alice's secret key $K_A$ is defined by 
\begin{equation}
d(\rho_{K_AE}|E)\le\varepsilon_{\rm sec},
\label{eq:secrecy_without_verification}
\end{equation}
where
\begin{eqnarray}
d(\rho_{K_AE}|E)
&:=&\frac12
\left\|
\rho_{K_AE}-\rho_{K_AE}^{\rm ideal}
\right\|_1,\\
\rho_{K_AE}^{\rm ideal}&=&2^{-\ell}\II_{K_A}\otimes\rho_E.
\end{eqnarray}
Furthermore, the protocol satisfies $\varepsilon_{\rm cor}$-correctness if the probability that Alice's and Bob's secret keys do not match is upper-bounded by $\varepsilon_{\rm cor}$, i.e., 
\begin{equation}
    \Pr[K_A\ne K_B]\le\varepsilon_{\rm cor}.
    \label{eq:correctness_without_verification}
\end{equation}
Under these conditions, the following lemma~\cite{Koashi_2009} holds.
\begin{lemma}
\label{lmm:separation_without_abortion}
(Separation lemma without error verification)
For QKD protocols without error verification, the trace distance is bounded as
\begin{equation}
\frac12
    \left\|
    \rho_{K_AK_BE}-\rho_{K_AK_BE}^{\rm ideal}
    \right\|_1
    \le
    d(\rho_{K_AE}|E)+
    \Pr[K_A\ne K_B].
    \label{eq:lemma1}
\end{equation}
That is, the security parameter $\varepsilon$ can be bounded as $\varepsilon\le\varepsilon_{\rm sec}+\varepsilon_{\rm cor}$.
\end{lemma}
Intuitively, this lemma means that if Alice's key is secret to Eve and matches Bob's key, then both Alice and Bob share a secret key.
We remark that due to Eve's attacks, the bit error rate can be increased at will. Therefore, in practice, it is impossible to ensure that $\Pr[K_A\ne K_B]$ in Eq.~(\ref{eq:lemma1}) is a small value.

\section{Separation lemma for QKD protocols with error verification\label{sec:separation_with_verification}}
In Sec.~\ref{sec:usual_separation}, we restricted ourselves to the case where decisions of continuing or aborting the protocol are made solely based on public information.
In practical QKD protocols, however, this restriction is often violated due to error verification~\cite{footnote_added}.

\subsection{Error verification's outcome must be announced}
\label{sec:Verifications_outcome_must_be_announced}
We first note that, in light of actual operations performed in QKD systems, it is unrealistic to assume that the outcome of error verification --- denoted by $v=0$ or $1$ for continuing or aborting the protocol --- can be kept permanently hidden from Eve. Therefore, it must be assumed that this information $v$ is always publicly available to Eve. This situation can be justified by the fact that the following scenario frequently occurs.
\begin{description}
    \item[Inevitable leakage of error verification's outcome] 
    Suppose, for example, that Alice and Bob execute a QKD protocol, and immediately after its completion, they use the generated secret key for secure communication with the one-time pad. In such a case, Eve can determine that the QKD protocol did not abort by observing a large volume of encrypted messages transmitted over the public channel. This implies that the outcome of error verification $v\in\{0,1\}$ is effectively leaked to Eve.
    \end{description}
In other words, even if Alice and Bob attempt to conceal $v\in\{0,1\}$ through encryption or other means, it is easy to construct scenarios in which the value of $v$ is leaked to Eve. Therefore, it is not reasonable to assume that $v$ remains concealed from Eve indefinitely, and it must instead be treated as publicly known.

\subsection{Separation lemma with error verification}
\label{sec:Separation_lemma_with_verification}

In order to describe variable $V$ properly, we use the following notation.
We treat $V$ as part of the public information accessible to Eve.
As in Sec.~\ref{sec:usual_separation}, we continue to let $\rho_{K_AK_BVE}$ denote the marginal state corresponding to the event where Alice and Bob decided to continue the protocol based solely on the public information.
In addition, we express the event where they decided to continue (or abort) due to error verification by
$\rho_{K_AK_BE}^{V=0}$ (or $\rho_{K_AK_BE}^{V=1}$), which is a marginal state of $\rho_{K_AK_BVE}$.
The final key state $\rho_{K_AK_BVE}$ then takes the form
\begin{eqnarray}
    \rho_{K_AK_BVE}&=&\sum_{v\in\{0,1\}}\rho_{K_AK_BE}^{V=v}\otimes[v]_V,
    \label{eq:decomposition_V}\\
    \rho_{K_AK_BE}^{V=1}&=&[\bot]_{K_A}\otimes[\bot]_{K_B}\otimes\rho_E^{V=1},
    \label{eq:lemma_2_V1_def}
\end{eqnarray}
with the symbol `$\bot$' denoting the situation where no key is generated since the verification failed.

In this notation, our observation of Sec.~\ref{sec:Verifications_outcome_must_be_announced} claims that it is inappropriate to evaluate the security using the left-hand side (LHS) of Eq.~(\ref{eq:def_security_without_veri}), where $V$ is not included as public information accessible to Eve. The security should rather be evaluated by the trace distance
\[
\frac12
\left\|\rho_{K_AK_BVE}-\rho_{K_AK_BVE}^{\rm ideal}\right\|_1, \]
for which the following separation lemma (a variant of Lemma \ref{lmm:separation_without_abortion}) holds.
\begin{lemma}
\label{lmm:separation_with_abortion}
(Separation lemma with or without error verification) 
For QKD protocols in general, with or without error verification, the security parameter can be upper-bounded as
\begin{equation}
\frac12
\left\|
    \rho_{K_AK_BVE}-\rho_{K_AK_BVE}^{\rm ideal}
    \right\|_1
\le
d\left(\rho_{K_AE}^{V=0}|E\right)+\Pr[K_A\ne K_B].
\label{eq:separation_general_case}
\end{equation}
\end{lemma}
\begin{proof}
By using Eqs.~(\ref{eq:decomposition_V}) and (\ref{eq:lemma_2_V1_def}), the trace distance with the ideal case can be bounded as
\begin{eqnarray}
\lefteqn{
\frac12
    \left\|\rho_{K_AK_BVE}-\rho_{K_AK_BVE}^{\rm ideal}\right\|_1}
    \nonumber\\
    &=&\frac12\sum_{v\in\{0,1\}}
    \left\|\rho_{K_AK_BE}^{V=v}-\left(\rho_{K_AK_BE}^{V=v}\right)^{\rm ideal}\right\|_1\nonumber\\
    &=&
    \frac12\left\|\rho_{K_AK_BE}^{V=0}-\left(\rho_{K_AK_BE}^{V=0}\right)^{\rm ideal}\right\|_1\nonumber\\
    &\le&  d\left(\rho_{K_AE}^{V=0}|E\right)+\Pr[K_A\ne K_B\land V=0]\nonumber\\
    &=&  d\left(\rho_{K_AE}^{V=0}|E\right)+\Pr[K_A\ne K_B].
    \label{eq:lemma_2_proof}
\end{eqnarray}
The first equality holds since the random variable $V$ is public. The second equality follows by the fact that $\rho_{K_AK_BE}^{V=1}$ is ideal, namely, $\rho_{K_AK_BE}^{V=1}=(\rho_{K_AK_BE}^{V=1})^{\rm ideal}$ (because no information is leaked to Eve when no key is generated), which can be seen from Eq.~(\ref{eq:lemma_2_V1_def}).
The inequality follows by applying Lemma~\ref{lmm:separation_without_abortion} to $\rho_{K_AK_BE}^{V=0}$.
The last equality holds since $\Pr[K_A\ne K_B\land V=1]=0$ due to Eq.~(\ref{eq:lemma_2_V1_def}).

Note that there is a practical method to upper-bound the second term $\Pr[K_A\ne K_B]$; see Appendix~\ref{sec:appendix1} for the detail.
\end{proof}

Comparing Lemmas~\ref{lmm:separation_without_abortion} and \ref{lmm:separation_with_abortion}, we observe that the quantity used to evaluate secrecy is replaced from $d(\rho_{K_AE}|E)$ to $d(\rho_{K_AE}^{V=0}|E)$. In other words, if we prove the security of QKD protocols with error verification, secrecy must be evaluated only with respect to the event conditioned on the success of error verification (i.e., $V=0$). 

\begin{description}
    \item[Secrecy condition with error verification] 
    The $\varepsilon_{\rm sec}$-secrecy, conditioned on the event that the verification succeeds (i.e., $V=0$), is expressed by 
    \begin{equation}
    d\left(\rho_{K_AE}^{V=0}|E\right)\le\varepsilon_{\rm sec}.
    \label{eq:Secrecy_with_verification}
    \end{equation}
\end{description}

Although many existing works based on Koashi's approach consider QKD protocols with error verification, they often adopt the LHS of Eq.~(\ref{eq:lemma1}) as the secrecy criterion~\cite{Takesue2015, PhysRevResearch.5.023132,PhysRevResearch.5.023065,Mizutani2015,Tamaki_2018,Curras2021,Curras2025}, rather than that of Eq.~(\ref{eq:separation_general_case}), which should be used to properly bound the trace distance in the presence of error verification \cite{footnote2}.
This indicates that the adopted definition is, in general, inadequate for QKD protocols with error verification. One might expect that the LHS of Eq.~(\ref{eq:separation_general_case}) can still be upper-bounded by the right-hand side (RHS) of Eq.~(\ref{eq:lemma1}). However, we will show in the next Sec.~\ref{subsec:conterexample} that this is not the case in general. Specifically, we demonstrate that, when error verification is present, there exists a situation in which the LHS of Eq.~(\ref{eq:separation_general_case}) cannot be bounded by the RHS of Eq.~(\ref{eq:lemma1}).

\subsection{Counterexample to bounding Eq.~(\ref{eq:separation_general_case}) by Eq.~(\ref{eq:lemma1})}
\label{subsec:conterexample}
In this section, we show by example that the LHS of Eq.~(\ref{eq:separation_general_case}) cannot, in general, be upper-bounded by the RHS of Eq.~(\ref{eq:lemma1}). 

In the following, the outcome of error verification is represented by a variable $V\in\{0,1\}$, which must be assumed known to Eve. 
More precisely, $V$ should be regarded not as a variable of Alice or Bob, but as the one accessible to Eve.

\paragraph{Protocol without error verification}
We assume that the reconciled key consists of two bits, with 
\begin{equation}
\rho_{ABE}=\frac18\sum_{x,y,z\in\{0,1\}}[xy]_A\otimes[zx]_B\otimes[z]_{E}.
\label{eq:counterexample_sifted}
\end{equation}
This corresponds, for example, to a situation in the BB84 protocol where Eve leaves the first qubit sent by Alice intact, performs the intercept-and-resend attack on the second qubit, swaps the two qubits, and then sends them to Bob.  

\begin{description}
\label{step:PA}
    \item[Privacy amplification (PA)] Alice and Bob set the first bit of the reconciled key as the secret keys $k_A,k_B$, namely, $k_A=a_1(=x)$, $k_B=b_1(=z)$.
\end{description}
In this case, the joint state of Alice's secret key and Eve's system is already the ideal state, as 
\begin{equation}
\rho_{K_AE}=\left(\frac12\II_2\right)_{K_A}\otimes\left(\frac12\II_2\right)_{E}
\end{equation}
holds. This means that 0-secrecy ($\varepsilon_{\rm sec}=0$) is satisfied, that is 
\begin{equation}
    d(\rho_{K_AE}|E)=0.
    \label{eq:secrecy_counterexample}
\end{equation}

\paragraph{Protocol with error verification added}
Suppose we add the following step to the above protocol.
\begin{description}
    \item[Error verification] 
    Bob compares his two reconciled key bits. If they match, the protocol proceeds; otherwise, Bob aborts the protocol.
\end{description}
This verification succeeds with probability 1/2, and the resulting (sub-normalized) state satisfies
\begin{align}
\rho_{K_AK_BEV}=&\frac14\left(\sum_{k\in\{0,1\}}[k]_{K_A}\otimes[k]_{K_B}\otimes[k]_{E}\otimes[0]_V\right)
\notag\\
&
+\frac14[\bot]_{K_A}\otimes[\bot]_{K_B}\otimes\II_E\otimes[1]_{V}.
\label{eq:rho_verified}
\end{align}
Clearly, 
\begin{equation}
\Pr\left[K_A\ne K_B\right]=0
    \label{eq:correctness_counterexample}
\end{equation}
holds, and the secret keys satisfy 0-correctness. 

To summarize, Eq.~(\ref{eq:secrecy_counterexample}) shows that $\varepsilon_{\rm sec}=0$, and as stated in Eq.~(\ref{eq:correctness_counterexample}), $\varepsilon_{\rm cor}=0$
also holds. Naively, one might therefore expect that combining these with Lemma~\ref{lmm:separation_without_abortion} would imply 0-security—that is, 
\begin{align}
    &\frac12\left\|\rho_{K_AK_BEV}-\left(\rho_{K_AK_BEV}\right)^{\rm ideal}\right\|_1
\notag\\
&\le
d(\rho_{K_AE}|E)+\Pr\left[K_A\ne K_B\right]=0.
\label{eq:wronginequality}
\end{align}
However, this is incorrect. In fact, a direct calculation shows that 
\begin{align}
\frac12\left\|\rho_{K_AK_BEV}-\left(\rho_{K_AK_BEV}\right)^{\rm ideal}\right\|_1=\frac14,
\label{eq:security_counterexample}
\end{align}
indicating that the actual situation is far from achieving 0-security.

\subsection{Analysis of the counterexample}
This section provides an analysis of the counterexample given in Sec.~\ref{subsec:conterexample}. If we evaluate secrecy using the inappropriate definition [Eq.~(\ref{eq:secrecy_without_verification})]—which should not be used for QKD protocols involving error verification—then, as shown in Eq.~(\ref{eq:secrecy_counterexample}), 0-secrecy appears to hold. However, when secrecy is assessed based on the correct definition [Eq.~(\ref{eq:Secrecy_with_verification})], we have
\begin{eqnarray}
d\left(\rho_{K_AE}^{V=0}|E\right)=\frac14,
\label{eq:secrecy_counterexample_with_verification}
\end{eqnarray}
which indicates that the state is far from satisfying 0-secrecy. We note that substituting Eqs.~(\ref{eq:correctness_counterexample}) and (\ref{eq:secrecy_counterexample_with_verification}) into Lemma~\ref{lmm:separation_with_abortion} yields a result consistent with Eq.~(\ref{eq:security_counterexample}). The fundamental reason for this discrepancy is that Eve gains additional information about Alice's secret key upon learning that the protocol has not been aborted (i.e., $V=0$). A more detailed explanation is given below.
\begin{itemize}
    \item 
    According to Eq.~(\ref{eq:counterexample_sifted}) and the verification procedure, the protocol ensures $k_A=E$ if $V=0$, and $k_A\neq E$ when $V=1$. 
    \item 
    In a protocol without error verification (i.e., where $v$ is not disclosed to Eve and the protocol is not aborted), Eve only has the information averaged over the above correlated ($k_A=E$) and anti-correlated events ($k_A\neq E$). As a result, the variable $k_A$ appears uniformly distributed, and 0-secrecy holds, as shown in Eq.~(\ref{eq:secrecy_counterexample}). 
    \item
    In contrast, for a protocol with error verification, the verification step succeeds with probability $1/2$, and its outcome is disclosed to Eve. In this case, Eq.~(\ref{eq:rho_verified}) implies that Alice's secret key is fully leaked to Eve, and secrecy can no longer be guaranteed. 
\end{itemize}

The counterexample above is a toy example indicative of what might happen in a real QKD protocol without error correction.
It illustrates an important point that the intuitive relation given by Eq.~(\ref{eq:wronginequality}) does not hold in general.

\section{Simple method for bounding secrecy parameter with error verification}
\label{sec:solution}

The counterexample in Sec.~\ref{subsec:conterexample} demonstrates that the variable $V$ can be correlated with the secret key. Consequently, even if secrecy were guaranteed in a situation where the key is generated without revealing $V$ (i.e., in a protocol without error verification), this does not necessarily imply security in the case where $V$ is made public. 
This discrepancy lies at the heart of the verification problem. 

However, for both mainstream methods of QKD security proofs, namely Renner's approach and Koashi's approach, we show in Secs.~\ref{sec:Solution_Renner} and \ref{sec:solutionkoashi}, respectively, that the verification problem can be resolved. In other words, in both approaches, if secrecy without revealing $V$, i.e., Eq.~(\ref{eq:secrecy_without_verification}), is guaranteed, then secrecy with the announcement of $V$, namely Eq.~(\ref{eq:Secrecy_with_verification}), can also be derived.

For simplicity of presentation, we will describe the case without smoothing. However, the same principles apply straightforwardly when smoothing is included.

\subsection{Solution in Renner's approach}
\label{sec:Solution_Renner}

When employing the Renner's approach for the security proof, Tomamichel and Leverrier have resolved the verification problem in~\cite{Tomamichel2017largelyself}.

\subsubsection{Setups and claims}
\label{sec:Renner_setting_claim}

We begin by explaining the setups.

First, in this section, we limit ourselves to the following type of error verification method: After bit error correction, Alice and Bob publicly announce classical information $H$, which is determined from their reconciled keys $A$ and $B$. Next, either Alice or Bob decides whether to continue or abort, disclosing the decision variable $V\in\{0,1\}$, based on $H$ and her (or his) reconciled key ($A$ or $B$). In other words, the public information $H$ and $V$ can be expressed by some functions $f$ and $g$ as $H=f(A,B)$ and $V=g(H,A)$ or $V=g(H, B)$.

Second, in this paper, by "Renner's approach" \cite{renner2006securityquantumkeydistribution} we always refer to the situation where (i) the protocol employs a universal$_2$ hash function (or more generally, the almost dual universal$_2$ function~\cite{TH2011,HT2013}) for privacy amplification, and (ii) the leftover hashing lemma (LHL) is used to prove the secrecy of the final key.

Under these setups, if the secrecy without revealing $V$, as in Eq.~(\ref{eq:secrecy_without_verification}), has been proven, then the secrecy with $V$ revealed, as in Eq.~(\ref{eq:Secrecy_with_verification}), automatically holds. 
In other words, among the two secrecy conditions [Eqs.~(\ref{eq:secrecy_without_verification}) and (\ref{eq:Secrecy_with_verification})], it suffices to prove only one of them.

\subsubsection{Mathematical details}
\label{sec:Renner_mathematical_detail}

Recall that, within Renner's approach, to prove the secrecy condition in Eq.~(\ref{eq:secrecy_without_verification}) without revealing $V$ as considered in Sec.~\ref{sec:usual_separation}, one usually discusses as follows. After bit error correction, only the public information $H$ is revealed, while $V$ remains hidden, and we consider the state $\rho_{ABEH}$. Based on the data obtained in the parameter estimation phase, one then proves that the conditional min-entropy satisfies 
\begin{equation}
    H_{\rm min}(A|EH)_\rho\ge \ell +2\log(1/\varepsilon_{\rm sec})
    \label{eq:H_min_lower_bound}.
\end{equation}
The secret key $K_A$ is obtained by applying privacy amplification to the reconciled key $A$. The state $\rho_{K_AEH}$ of this secret key can then be shown, by the LHL together with Eq.~(\ref{eq:H_min_lower_bound}), to satisfy 
\begin{equation}
    d\left(\rho_{K_AEH}|EH\right)
    \le 2^{\frac12\left(\ell-H_\text{min}(A|EH)_{\rho}\right)}.
    \label{eq:LHL_without_ver}
\end{equation}
Thus, Eq.~(\ref{eq:secrecy_without_verification}) is establied.

Next, we evaluate the secrecy condition in Eq.~(\ref{eq:Secrecy_with_verification}) when $V$ is revealed. This corresponds, by definition, to deriving an upper bound on $d(\rho_{K_AEH}^{V=0}|EH)$.

To this end, let us first note the following. After the state $\rho_{ABEH}$ is generated as described two paragraphs earlier, Alice and Bob compute and reveal $V\in\{0,1\}$, and denote the resulting state by $\rho_{ABEHV}$. In this case, the following statements hold. 
\begin{itemize}
\item 
If Alice applies privacy amplification to the reconciled key $A$ of $\rho_{ABEH}=\text{tr}_V\left(\rho_{ABEHV}\right)$, one can reproduce the sub-normalized state $\rho_{K_AEH}$, which is employed in the evaluation of the secrecy considered without revealing $V$, in Eq.~(\ref{eq:secrecy_without_verification}).
\item 
Consider the sub-normalized state $\rho_{ABEH}^{V=0}$ obtained by projecting $\rho_{ABEHV}$ onto the case $V=0$, i.e., $\rho_{ABEH}^{V=0}=\text{tr}_V\left(\rho_{ABEHV}[0]_V\right)$. 
If one then applies privacy amplification to the reconciled key $A$, the resulting state coincides with the sub-normalized state $\rho_{K_AEH}^{V=0}$ used in the evaluation of the secrecy condition in Eq.~(\ref{eq:Secrecy_with_verification}). 
\end{itemize}

In summary, to evaluate the secrecy condition without revealing $V$, it suffices to apply the LHL using the conditional min-entropy $H_\text{min}(A|EH)_\rho$ of $\rho_{ABEH}$. On the other hand, to evaluate the secrecy condition when $V$ is revealed, one should use the conditional min-entropy $H_\text{min}(A|EH)_{\rho^{V=0}}$ of $\rho_{ABEH}^{V=0}$. It is known that the following relation holds between these two conditional min-entropies. 

\begin{lemma}[Ref. \cite{Tomamichel2017largelyself}, Lemma 10]
\label{lmm:H_min_non_decreasing}
The conditional min-entropy of (possibly sub-normalized) sate $\rho$ does not decrease when marginalized by the condition $V=0$, i.e.,
\begin{equation}
        H_{\rm min}(A|EH)_{\rho}\le H_{\rm min}(A|EH)_{\rho^{V=0}}.
        \label{eq:H_min_with_ver}
\end{equation}
\end{lemma}

Thanks to this lemma, as long as Eq.~(\ref{eq:H_min_lower_bound}) holds, 
\begin{equation}
H_{\rm min}(A|EH)_{\rho^{V=0}}\ge \ell +2\log(1/\varepsilon_{\rm sec})   
\end{equation}
is satisfied. By applying the LHL to $\rho^{V=0}$, we obtain
\begin{eqnarray}
    d\left(\rho_{K_AEH}^{V=0}|EH\right)&\le&2^{\frac12\left(\ell-H_\text{min}(A|EH)_{\rho^{V=0}}\right)},
\end{eqnarray}
which shows that Eq.~(\ref{eq:Secrecy_with_verification}) is fulfilled.

\subsubsection{Proof of Lemma \ref{lmm:H_min_non_decreasing}}

The proof of Lemma~\ref{lmm:H_min_non_decreasing} is given in Ref.~\cite{Tomamichel2017largelyself}, but for the reader's convenience, we provide it here. 

Since the projection $[0]_V$ on the space $V$ is a quantum operation, 
\[
[0]_V\rho_{ABEHV}[0]_V\le \rho_{ABEHV}
\]
holds. By tracing out subsystems $B$ and $V$, we obtain
\[
\rho_{AEH}^{V=0}\le \rho_{AEH}.
\]
By combining this with the definition of the conditional min-entropy, we have Eq.~(\ref{eq:H_min_with_ver}).

\subsection{Solution in Koashi's approach}
\label{sec:solutionkoashi}

In security proofs based on the phase-error-correction method (the PEC-based approach, also known as Koashi's approach~\cite{koashi2005simple,Koashi_2009}), no general solution to this verification problem has been known. However, in this section, we provide such a solution.



\subsubsection{Summary of Koashi's approach and the challenge of addressing the verification problem within this framework}
\label{sec:summarykoashidifficulty}

Recall that security proofs in Koashi's approach usually proceed as follows (see, e.g., Refs. \cite{Tsurumaru_2020,9606732}).
\begin{enumerate}
\item[i)] Define a virtual pure state $\ket{\rho}_{ABE}$ which equals $\rho_{AE}$ when subsystem $A$ is measured in the bit basis (usually chosen to be the $Z$ basis) and $B$ is traced out \footnote{
Note that subsystem $B$ does not only consist of Bob's system, but also includes any system accessible to Alice. For example, it includes subsystems other than the qubits possessed by Alice, such as the shield system (i.e., the system that purifies Alice's emitted states).
}.
\item[ii)] 
Let $\rho_{X^AB}$ be the {\it virtual} state, obtained by measuring subsystem $A$ of $\ket{\rho}_{ABE}$ in the phase basis (usually chosen to be the $X$ basis) and tracing out $E$.
Upper-bound the conditional max-entropy $H_\text{max}(X^A|B)_\rho$, using the data obtained in the parameter estimation phase.
\item[iii)]
Suppose that one performs error correction on subsystem $X^A$ in the phase basis, using $B$ as side information.
Use $H_\text{max}(X^A|B)_\rho$ to obtain an upper bound $Q^\text{EC}$ on the failure probability of the above phase error correction \footnote{
For example, Ref. \cite{Koenig2009} states that $Q^\text{EC}\le 2^{H_\text{max}(X^A|B)_\rho-(n-\ell)}$.
}.
Then the secrecy of Alice's secret key $K_A$ can be given as $d(\rho_{K_AE}|E)\le2\sqrt2\sqrt{Q^\text{EC}}$; see, e.g., Refs. \cite{Hayashi_2012,Tsurumaru_2020,9606732}.
In other words, the secrecy parameter can be bounded as $\varepsilon_\text{sec}\le2\sqrt2\sqrt{Q^\text{EC}}$.
\end{enumerate}

As described above, in Koashi's approach, the target state of the security proof is not the state $\rho_{ABE}$ corresponding to the actual QKD protocol, but rather the virtual state $\rho_{X^AB}$ that is mathematically defined from $\rho_{ABE}$. In this framework, Alice's sifted key $A$ is not defined; instead, the state obtained after measuring in the phase basis is considered. 
Consequently, the description of the public information $H$ and $V$ is not straightforward (in contrast, in Renner's approach discussed in the previous section~\ref{sec:Solution_Renner}, the reconciled keys $A$ and $B$ appear explicitly as classical variables in the state $\rho_{ABE}$ after error correction but before the calculation of $H$ and $V$, so that the state $\rho_{ABEHV}$ including $H$ and $V$ can be described straightforwardly). This has made it difficult to address the verification problem in Koashi's approach. For example, within this approach it is not clear whether a lemma analogous to Lemma~\ref{lmm:H_min_non_decreasing} exists. For these reasons, when adopting a security proof based on Koashi's approach, no general solution to the verification problem has been known.

\subsubsection{Proposed solution}
\label{sec:proposedsolution}

The idea of the solution is to exploit the fact that Step iii) in the previous section~\ref{sec:summarykoashidifficulty} is equivalent to the LHL in Renner's approach~\cite{Tsurumaru_2020,9606732}. Using this equivalence, we translate the situation of Step~iii) into Renner's approach, and then apply the solution described in the previous section~\ref{sec:Solution_Renner} for Renner's approach.

We begin by stating the conclusion, and the mathematical details will be given in the next section.

Our conclusion is the following: Suppose Koashi's approach is applied to a protocol without aborting due to error verification, in the same sense as in Sec.~\ref{sec:usual_separation}. Further assume that in Step ii) we obtain the following upper bound on the conidtional max-entropy
\begin{equation}
H_{\rm max}(X^A|B)_\rho\le H_\text{max}^\text{th},
\label{eq:H_max_upper_bound}
\end{equation}
where $H_\text{max}^\text{th}$ is a constant once the parameter estimation phase is completed. If we then add the error-verification procedure of Sec.~\ref{sec:Renner_setting_claim} to the protocol, the secrecy condition with aborting due to error verification: 
\begin{equation}
    d(\rho_{K_AEH}^{V=0})\le2^{\frac12\left(\ell-n+H_\text{max}^\text{th}+|H|\right)}
    \label{eq:upperbound_d_Koashi}
\end{equation}
holds. Here, $|H|$ denotes the bit length of $H$.

\subsubsection{Mathematical details}

Once Eq.~(\ref{eq:H_max_upper_bound}) holds, by an entropic uncertainty relation~\cite{PhysRevLett.106.110506}, we obtain a lower bound on the min-entropy $H_{\min}(A|E)\ge n-H_\text{max}^\text{th}$, and 
by the chain rule of the conditional min-entropy, Eq.~(3.21) in~\cite{renner2006securityquantumkeydistribution}, we have
\begin{equation}
H_{\min}(A|EH)\ge n-H_\text{max}^\text{th}-|H|.  
\end{equation}
This lower bound can be identified with Eq.~(\ref{eq:H_min_lower_bound}), by letting $\ell=n-H_\text{max}^\text{th}-|H|-2\log(1/\varepsilon_\text{sec})$. 
With this, we have completed the translation of Koashi's approach into Renner's approach in Sec.~\ref{sec:Renner_mathematical_detail}. By applying a solution analogous to that in Sec.~\ref{sec:Renner_mathematical_detail}, Eq.~(\ref{eq:upperbound_d_Koashi}) can then be established.

\ 

\section{Discussion\label{sec:discussion}}

The verification problem identified in this paper originates from the fact that the verification's outcome $V$ can, in general, be correlated with the sifted, reconciled or final keys. On the other hand, it should be noted that if one can somehow prove that $V$ is uncorrelated with the keys, then this issue does not arise. As already discussed in Sec.~\ref{sec:usual_separation}, such a situation occurs, for example, when the decisions to continue or abort the protocol are made solely based on the public information.

We also note that there is another typical situation where $V$ can be shown uncorrelated with the keys.
That is where one can apply a Shor–Preskill-type security proof~\cite{PhysRevLett.85.441}, and thus regard the error verification step as part of the syndrome measurement for bit error correction (or, equivalently, if it is incorporated into the choice of a sufficiently large code $C_1$ for $Z$-basis error correction).
This is true, for example, when Alice and Bob can be assumed to possess qubits in the virtual protocol (as in the PEC-based approach) and perform error verification using a linear hash function~\cite{footnote4}.
In such cases, the secrecy of Alice's (or Bob's) final key can be discussed independently of error verification, and thus the verification problem no longer occurs~\cite{footnote5}.

\section*{Acknowledgements}
We thank Kiyoshi Tamaki, Go Kato and Shun Kawakami for helpful discussions. 
Also, we are deeply grateful to the anonymous referee for pointing out that, in Renner's approach, the situation where the outcome of error verification is disclosed to Eve has already been incorporated into the security proof in Ref.~\cite{Tomamichel2017largelyself}.
A. Mizutani is partially supported by JSPS KAKENHI Grant Number JP24K16977.

\appendix

\section{Practical method for bounding $\Pr[K_A\ne K_B]$}
\label{sec:appendix1}
There is a practical method for bounding the probability $\Pr[K_A\ne K_B]$ appearing, e.g., in Eqs.~(\ref{eq:correctness_without_verification}) and (\ref{eq:separation_general_case})~\cite{Tomamichel2017largelyself}.
This is the probability of an undesirable event in which the secret keys do not match despite the error verification being successful. This probability can be upper-bounded as
\begin{align}
&\Pr[K_A\ne K_B]=\Pr[K_A\ne K_B\land V=0]\notag
\\
&\le\Pr[A\ne B\land V=0]=\Pr[V=0\,|\, A\ne B]\Pr[A\ne B]
\notag
\\
&\le\Pr[V=0\,|\, A\ne B].
\end{align}
Here, $A$ and $B$ denote Alice's and Bob's reconciled keys, respectively. The quantity on the last line (and thus also $\Pr[K_A\ne K_B]$) can be upper-bounded by $\varepsilon_{\rm cor}$ as follows.
Suppose that Alice announces the hash value $h(a)$ of her reconciled key $a$, using a randomly chosen element $h$ of the universal hash function $H$ with the output length $\lceil\log(1/\varepsilon_{\rm cor})\rceil$.
Also, suppose that Bob announces that the protocol is aborted ($v=1$) if and only if the hash values of the reconciled keys differ, i.e., $h(a)\ne h(b)$. Then, we have 
\begin{equation}
\Pr[V=0\,|\, A\ne B]=\Pr[H(A)=H(B)\,|\, A\ne B]\le \varepsilon_{\rm cor}.
\label{eq:correctness_verification}
\end{equation}


\begin{thebibliography}{99}
\bibitem{footnote_added}
Note that, as mentioned in Sec.~\ref{sec:usual_separation}, when the abort decision is determined solely from the public information (for example, when the QBER is disclosed and the protocol is aborted whenever it exceeds a predetermined threshold), then even if the abort/non-abort outcome is publicly announced, no additional sacrifice in privacy amplification is required. In contrast, in the error-verification procedure considered from Sec.~\ref{sec:separation_with_verification}, Alice announces a hash value whereas Bob does not; instead, Bob checks consistency and announces the verification outcome. In this case, the abort decision is not determined solely from the public information, and in principle one must take into account the possibility of further shortening the key during privacy amplification. 
Nevertheless, as explained in Sec.~\ref{sec:solution}, within the security proof based on the leftover hashing lemma, no additional sacrifice is required.
\bibitem{footnote2}
We note, however, that all of them can be corrected without increasing the amount of privacy amplification, as justified in Sec.~\ref{sec:solutionkoashi} of this paper.

\if0
\bibitem{footnote5}
Reference~\cite{Hayashi_2012} fits this situation if one restricts that the universal hash function used for error verification, mentioned at the end of Section 2.2, to be linear (though we failed to say that it should be linear).
\end{thebibliography}
\end{document}